\documentclass[11pt]{article}
\usepackage[utf8]{inputenc}
\usepackage[british]{babel}
\usepackage{lmodern}
\usepackage{comment}
\usepackage{amsthm,amssymb,amsmath}
\usepackage{mathtools}
\usepackage{graphicx,algorithmic,algorithm}
\usepackage{enumerate, enumitem}
\usepackage{thmtools, thm-restate}

\usepackage{tikz,tikz-cd}
\usepackage{circuitikz}
\tikzcdset{scale cd/.style={every label/.append style={scale=#1},
    cells={nodes={scale=#1}}}}%tikzcd size 조절
\usetikzlibrary{decorations.markings}
\usetikzlibrary{shapes}
\usepackage[letterpaper, margin=1in]{geometry}
\usepackage[colorlinks=true]{hyperref}
\usepackage[nameinlink,capitalize]{cleveref}
\usepackage{todonotes}
\usepackage{dsfont}
\usepackage[normalem]{ulem}
\hypersetup{
	pdftitle=   {},
	pdfauthor=  {}
}
\usepackage{authblk}
\usepackage{caption}
\usepackage{subcaption}
\usepackage{commath}

\newtheorem{theorem}{Theorem}[section]
\newtheorem{lemma}[theorem]{Lemma}
\newtheorem{proposition}[theorem]{Proposition}

\newtheorem{observation}[theorem]{Observation}
\newtheorem{claim}[theorem]{Claim}
\newtheorem{question}{Question}

\newenvironment{clproof}{\begin{list}{}{%
			\setlength{\leftmargin}{3mm}%
		} \item {\it Proof.} }{\hfill$\lozenge$\end{list}}

\newcommand\extrafootertext[1]{%
    \bgroup
    \renewcommand\thefootnote{\fnsymbol{footnote}}%
    \renewcommand\thempfootnote{\fnsymbol{mpfootnote}}%
    \footnotetext[0]{#1}%
    \egroup
}

\newcommand{\yes}{\textsc{Yes}}

\newcommand{\cw}{\mathsf{cw}}

\newcommand{\depth}{\mathsf{depth}}

\newcommand{\crank}{\mathsf{cr}}

\newcommand{\dpw}{\mathsf{dpw}}

\newcommand{\dcw}{\mathsf{dcw}}
\newcommand{\lab}{\mathsf{lab}}
\newcommand{\DCW}{\mathsf{DCW}}
\newcommand{\op}{\mathsf{op}}

\newcommand{\DP}{\mathsf{DP}}

\renewcommand\abs[1]{\left\lvert #1\right\rvert}

\begin{document}

	\title{An FPT algorithm for cycle rank on semi-complete digraphs}
	
	\author[1,2]{Seokbeom Kim}
	\author[2,3]{O-joung Kwon}
	\author[3]{Myounghwan Lee}

    \affil[1]{Department of Mathematical Sciences, KAIST, Daejeon, South Korea}
	\affil[2]{Discrete Mathematics Group, Institute for Basic Science (IBS), Daejeon, South Korea}
	\affil[3]{Department of Mathematics, Hanyang University, Seoul, South Korea.}
	
	\date\today
	\maketitle

	\extrafootertext{S.~Kim and O. Kwon are supported by the Institute for Basic Science (IBS-R029-C1).
    O.~Kwon and M.~Lee are supported by the National Research Foundation of Korea (NRF) grant funded by the Ministry of Science and ICT (No. RS-2023-00211670).
    }

    \extrafootertext{E-mail addresses: 
  \texttt{seokbeom@kaist.ac.kr} (S.~Kim),
  \texttt{ojoungkwon@hanyang.ac.kr} (O.~Kwon), and \texttt{sycuel@hanyang.ac.kr} (M.~Lee).}

\begin{abstract}
Cycle rank is a depth parameter for digraphs introduced by Eggan in 1963. Gruber (DMTCS 2012) and Giannopoulou, Hunter, and Thilikos (DAM 2012) asked whether the problem of determining if a given digraph has cycle rank at most $w$ is fixed-parameter tractable parameterized by $w$. 
We provide such algorithms for semi-complete digraphs, and for digraphs of bounded directed clique-width. Specifically, we show that given an $n$-vertex semi-complete digraph~$G$ and an integer $w$, one can in time $2^{\mathcal{O}((w+1)4^{w})} n^2$ determine whether $G$ has cycle rank at most~$w$. The proof is reduced to the case of bounded directed clique-width, and we then show that given an $n$-vertex digraph $G$ with a directed clique-width $k$-expression and an integer $w$, one can in time $2^{\mathcal{O}((w+1)4^k)}n$ determine whether $G$ has cycle rank at most $w$.
Additionally, we consider the \textsc{Minimum Feedback Arc Set} problem on semi-complete digraphs, and show that it can be solved in time $n^{\mathcal{O}(w)}$, when a cycle rank decomposition of depth $w\ge 1$ is given.
\end{abstract}

\section{Introduction}

While undirected graph width parameters such as treewidth~\cite{RobertsonSeymour86twalgorithm}, tree-depth~\cite{NesetrilO2006}, and clique-width~\cite{Courcelle2000lineartimesolvable, Courcelle00upperboundcw} have become standard tools, the situation for digraphs is considerably more subtle.
Numerous directed width parameters have been proposed, including directed treewidth~\cite{JohnsonRST2001}, DAG-width~\cite{BerwangerDHK2006DAG, obdrvzalek06DAGwidth}, and Kelly-width~\cite{HUNTER2008206}; see a survey of Kreutzer and Kwon~\cite{Kreutzer2018} on digraph width parameters.
Ganian, Hlin\v{e}n\'{y}, Kneis, Meister, Obdr\v{z}\'{a}lek, Rossmanith, and Sikdar~\cite{GanianHKMORS2016} discussed that there are inherent obstacles to obtaining a directed analogue that simultaneously enjoys the full range of desirable algorithmic and structural properties. Consequently, different width parameters have been developed for different purposes.

Within this landscape, the notion of \emph{cycle rank} stands out as an interesting parameter, which roughly measures the complexity of dismantling all strongly connected components of a digraph (see~\Cref{subsection:cyclerank} for the formal definition). 
Originally, Eggan~\cite{Eggan1963} introduced cycle rank in the course of studying the star-height of regular expressions in formal language theory.
Since it is reminiscent of the \emph{tree-depth} of undirected graphs~\cite{NesetrilO2006}, cycle rank has attracted attention from both structural and algorithmic perspectives.

From a structural point of view, Giannopoulou, Hunter, and Thilikos~\cite{GIANNOPOULOU2012searchinggame} characterized digraphs of bounded cycle rank by providing obstructions called \emph{LIFO-haven} and \emph{directed shelter}. Furthermore, Hatzel, Kwon, Lee, and Wiederrecht~\cite{hatzel25+unavoidable} recently provided three types of unavoidable butterfly-minors for digraphs of large cycle rank.
For the algorithmic side, Ganian, Hlin\v{e}n\'{y}, Kneis, Langer, Obdr\v{z}\'{a}lek, and Rossmanith~\cite{Ganian2014} summarized previously known results and established new findings regarding the parameterized complexity of several \textsf{NP}-hard problems on digraphs when parameterized by cycle rank.

Despite considerable progress in the study of cycle rank, its computational intractability remains a significant obstacle.
Gruber~\cite{Gruber2012} showed that computing the cycle rank of a digraph is \textsf{NP}-hard, even when the input is restricted to digraphs of maximum out-degree at most~$2$.
As a mitigating result, for $w\ge 1$, Giannopoulou, Hunter, and Thilikos~\cite{GIANNOPOULOU2012searchinggame} provided an $\mathcal{O}(n^{w})$-time algorithm that decides if an $n$-vertex digraph has cycle rank at most $w$.
Based on these results, it is natural to ask whether there is a fixed-parameter tractable algorithm for deciding whether the cycle rank of a digraph is at most $w$, that is, an algorithm running in time $f(w) \cdot n^{\mathcal{O}(1)}$ for some function $f$. 
This question was posed by Gruber~\cite{Gruber2012} and by Giannopoulou, Hunter, and Thilikos~\cite{GIANNOPOULOU2012searchinggame}. 

In parameterized complexity, an instance of a parameterized problem consists in a pair $(x,k)$, where $k$ is a secondary measurement, called the parameter.
A parameterized problem $Q \subseteq \Sigma^* \times \mathbb{N}$ is \emph{fixed-parameter tractable (FPT)} if there is an algorithm which decides whether $(x,k)$ belongs to $Q$ in time $f(k)\cdot \abs{x}^{\mathcal{O}(1)}$ for some computable function $f$.
Such an algorithm is called a \emph{fixed-parameter tractable algorithm (FPT algorithm)}.
A parameterized problem $Q \subseteq \Sigma^* \times \mathbb{N}$ is \emph{slicewise polynomial (XP)} if there is an algorithm which decides whether $(x,k)$ belongs to $Q$ in time $\abs{x}^{f(k)}$ for some computable function $f$.
Such an algorithm is called a \emph{slicewise polynomial algorithm (XP algorithm)}.

\paragraph{Our contribution.}

We focus on the decision problem.
\vskip 0.3cm
\noindent
\fbox{\parbox{0.97\textwidth}{
	\textsc{Cycle Rank}\\
	\textbf{Input:} A digraph $G$ and a nonnegative integer $w$\\
	\textbf{Task:} Determine whether $G$ has cycle rank at most $w$.
    }}
\vskip 0.3cm

The existence of a fixed-parameter tractable algorithm for \textsc{Cycle Rank} remains an open problem. In this paper, we study the problem on \emph{semi-complete} digraphs.
Semi-complete digraphs are digraphs obtained from complete undirected graphs by replacing every edge with one directed edge or bi-directed edges. The class of semi-complete digraphs is one of the most extensively studied classes of digraphs and admits many deep structural and algorithmic results; see a survey of Bang-Jensen and Havet~\cite{BangJensen2018}. 
For further results on the width parameters for semi-complete digraphs, we refer to~\cite{FominP2013,fominpilipczuk19dpw,FradkinS2013,GurskiKRW2021,Pilipczuk2013}.

Our main result is a fixed-parameter tractable algorithm for semi-complete digraphs. We denote by $\crank(G)$ the cycle rank of a digraph $G$.

\begin{restatable}{theorem}{mainonethm}\label{thm:main1}
There is an algorithm that, given a semi-complete digraph $G$ on $n$ vertices and an integer $w \geq 0$, decides if $\crank(G)\le w$ in time~$2^{\mathcal{O}((w+1)4^{w})}\cdot n^2$.
Moreover, one can find a cycle rank decomposition of $G$ of depth at most $w$ in the same time if $\crank(G)\le w$.
\end{restatable}

\noindent
Running the algorithm in \cref{thm:main1} from $w=0$ to $w=\crank(G)+1$, we can derive that 
the cycle rank of a semi-complete digraph $G$ can be computed exactly in time $2^{\mathcal{O}((\crank(G)+2)4^{\crank(G)})}\cdot n^2$.

We explain the basic strategy for~\cref{thm:main1}.
Gruber~\cite{Gruber2012} showed that the directed path-width of a digraph is bounded from above by its cycle rank. Moreover, Fomin and Pilipczuk~\cite{fominpilipczuk19dpw} obtained a fixed-parameter tractable algorithm for testing whether a given semi-complete digraph has directed path-width at most $w$, and, if so, produces a corresponding decomposition.
Furthermore, Fomin and Pilipczuk~\cite{FominP2013} showed that given a directed path-decomposition of a semi-complete digraph of width $w$, one can obtain a directed (clique-width) $(w+2)$-expression in time $\mathcal{O}(n^2)$. Combining these results, we may assume that a given semi-complete digraph is equipped with a directed $(w+2)$-expression.

We remark that directed clique-width was introduced by Courcelle and Olariu~\cite{Courcelle00upperboundcw}, and it is a variant of undirected clique-width in which, instead of an operation that adds undirected edges between two sets of vertices with given labels, we have an operation that adds directed edges from the set of vertices with label $i$ to the set of vertices with label $j$. 

Motivated by this, we design a fixed-parameter tractable algorithm for general digraphs of bounded directed clique-width.
As discussed, \cref{thm:main} implies \cref{thm:main1}.

\begin{restatable}{theorem}{mainthm}\label{thm:main}
    Let $k \geq 1$ be an integer.
    There is an algorithm that, given an $n$-vertex digraph $G$, a directed $k$-expression of $G$, and an integer $w \geq 0$, decides if $\crank(G)\le w$ in time~$2^{\mathcal{O}((w+1)4^k)}n$.
    Moreover, one can find a cycle rank decomposition of $G$ of depth at most $w$ in time~$2^{\mathcal{O}((w+1)4^k)}(n+m)$ if $\crank(G)\le w$, where $m\coloneqq \lvert E(G)\rvert$.
\end{restatable}

We remark that the property of having cycle rank at most $w$ can be expressed by an MSO$_1$-formula. 
Hence, the existence of a fixed-parameter tractable algorithm parameterized by $k$ and $w$, where $k$ and $w$ are as in \Cref{thm:main}, follows from the meta-theorem for bounded directed clique-width classes~\cite[Theorem 4.2]{Ganian2014}. However, our algorithm is explicit and has a significantly better running time.

Kant\'e and Rao~\cite{KanteR2013} showed that there is a fixed-parameter tractable algorithm for deciding whether the bi-rank-width of a digraph is at most $k$.
Using a relation between bi-rank-width and directed clique-width, one can obtain an algorithm that runs in time $g(k) \cdot n^3$ for some function $g$ and either determines that the directed clique-width of a digraph is more than $k$, or outputs a directed clique-width expression of width at most $2^{2k+2}-1$~\cite[Lemma 9.9.15]{Kreutzer2018}.
Therefore, \cref{thm:main} implies that \textsc{Cycle Rank} is fixed-parameter tractable parameterized by $w$ and directed clique-width.

Steiner and Wiederrecht~\cite{SteinerW2020} proved that \textsc{Cycle Rank} is fixed-parameter tractable parameterized by the \emph{directed modular width}. 
Their algorithm heavily relies on the fact that an optimal cycle rank decomposition can be nicely decomposed along the modules of the digraph.
As this property is quite restrictive, their algorithm does not readily extend to classes of bounded directed clique-width. 
In addition, the semi-complete digraph obtained from a directed path $v_1v_2 \cdots v_k$ $(k\ge 4)$ by adding all edges $v_iv_j$ with $i>j$ has directed clique-width $3$ and directed modular width $k$. Thus, their algorithm does not imply \cref{thm:main} even on semi-complete digraphs.

A main tool in the proof of~\Cref{thm:main} is an equivalent coloring concept, where every strongly connected subdigraph has a unique vertex with largest color.
This leads to a new ranking parameter, called the \emph{CW-ranking number} (the abbreviation ``CW'' stands for closed walks).
It is motivated by the \emph{vertex ranking number} of undirected graphs, introduced by Iyer, Ratliff, and Vijayan~\cite{IyerRV1988}, which is known to be equivalent to the \emph{tree-depth} of an undirected graph~\cite{NesetrilO2006, NesetrilO2012}. This is also closely related to the \emph{centered coloring} of graphs~\cite{NesetrilO2012}.

For a digraph $G$ and an integer $k \geq 0$, a function $h : V(G) \to [0, k]$ is a \emph{CW-ranking with $k$ ranks} if every closed directed walk $C$ of $G$ has a unique vertex $v$ for which $h(v) = \max_{w \in V(C)} h(w)$.
The \emph{CW-ranking number} of a digraph $G$, denoted by $\chi_\cw(G)$, is the minimum $k$ such that $G$ admits a CW-ranking with $k$ ranks.
We prove in~\Cref{CW_ranking_equivalence} that the cycle rank of a digraph is equal to its CW-ranking number. 

At each node of a directed clique-width expression tree, we need to store the possible CW-rankings of the corresponding subdigraph. At a later node, we need to determine whether a new closed walk containing two vertices of maximum rank is created. To this end, we store certain reachability information for vertices. More precisely, for label sets $A$ and $B$, we count the number of vertices $v$ such that $A$ is exactly the set of labels $a$ for which there exists a directed path from a vertex of label $a$ to $v$, and $B$ is exactly the set of labels $b$ for which there exists a directed path from $v$ to a vertex of label $b$. If there are two such vertices, then adding all edges from vertices of label in $B$ to vertices of label in $A$ creates a bad closed walk. Therefore, it is important to keep track of this information, and we do so using dynamic programming.

We remark that an analogous result is known for undirected tree-depth. Reidl, Rossmanith, Vilaamil, and Sikdar~\cite{ReidlRVS2014} showed that the tree-depth of a graph $G$ can be computed exactly in time $2^{\mathcal{O}(\mathsf{td}(G)^2)}n$.
In the course of proving this result, the authors showed that given a tree decomposition of the graph of width $w$ and an integer $t$, one can determine whether the tree-depth of the graph is at most $t$ in time $2^{\mathcal{O}(wt)}n$. 
A well-known open problem is whether there exists a constant-factor approximation algorithm running in time $2^{o(\mathsf{td}(G)^2)}\cdot n^c$ for any constant $c\in\mathbb{N}$; see~\cite{pilipczuk2025}.

Our FPT algorithm could potentially be used to show that certain algorithmic problems on semi-complete digraphs are FPT parameterized by cycle rank. If a problem is MSO$_1$-expressible, then one can use the meta-theorem for directed clique-width; otherwise, establishing such a result may be nontrivial. 
We could not find a relevant result of this kind. 
But we initiate the study of a natural problem on semi-complete digraphs: \(\textsc{Feedback Arc Set}\) parameterized by cycle rank. 
The problem parameterized by the solution size has been studied in the literature~\cite{AlonLS2009, ChenLLOR2008, MisraSSZ2023}.

\paragraph{\textsc{Minimum Feedback Arc Set} on semi-complete digraphs.}
We investigate the \textsc{Minimum Feedback Arc Set} problem on semi-complete digraphs.
The problem asks for a minimum number of edges whose removal makes the input digraph acyclic.
While \textsc{Minimum Feedback Arc Set} on general digraphs is one of Karp's 21 \textsf{NP}-complete problems~\cite{Karp1972}, its \textsf{NP}-completeness on tournaments was conjectured by Bang-Jensen and Thomassen~\cite{BangJensenT1992}.
This conjecture was solved independently by Alon~\cite{Alon2006}, Charbit, Thomass{\'e}, and Yeo~\cite{CharbitTY2007}, and Conitzer~\cite{Conitzer2006}, who proved that \textsc{Minimum Feedback Arc Set} remains \textsf{NP}-complete even on tournaments, and consequently on semi-complete digraphs. Thus, it is natural to consider the problem parameterized by structural parameters.

We focus on the question of whether \textsc{Minimum Feedback Arc Set} on
semi-complete digraphs is fixed-parameter tractable parameterized by cycle rank. Although we do not resolve this question, we observe that the problem is XP with respect to this parameter.

\begin{restatable}{theorem}{thmMFAS}\label{thm:MFAS}
Given an $n$-vertex semi-complete digraph and its cycle rank decomposition of depth $w\ge 1$, \textsc{Minimum Feedback Arc Set} can be solved in time $n^{\mathcal{O}(w)}$.
\end{restatable}

We note that, unlike the \textsc{Minimum Directed Feedback Vertex Set} problem, finding a minimum feedback arc set requires optimizing over edge subsets.
This distinction becomes critical when considering the LinEMSO$_1$ optimization problem (see \cite{Courcelle2000lineartimesolvable} for the definition of LinEMSO$_1$ optimization problem).
By the definition of a LinEMSO$_1$ optimization problem, any evaluation function for such a problem minimizes a linear combination of the sizes of vertex subsets, implying its optimal value is bounded by $\mathcal{O}(n)$ on unweighted $n$-vertex digraphs.
However, the size of a minimum feedback arc set can be $\Omega(n^2)$ even on tournaments~\cite[Theorem~2.9.5]{BangJensen2018}.
Due to this asymptotic difference, the problem cannot be expressed as a LinEMSO$_1$ optimization problem.
Consequently, the tractability of \textsc{Minimum Feedback Arc Set} on semi-complete digraphs does not trivially follow from the meta-theorem on bounded directed clique-width classes~\cite[Theorem 4.2]{Ganian2014}.

\paragraph{Organization.}
The rest of the paper is organized as follows.
In \cref{sec:prelim}, we provide some notations and definitions.
In \cref{sec:CW-ranking}, we introduce the concept of CW-rankings and prove its equivalence to the cycle rank.
In \cref{sec:bdddcw}, we show that deciding the cycle rank of a given digraph of bounded directed clique-width is fixed-parameter tractable (\cref{thm:main}).
In \cref{sec:semicomplete}, we show the analogous result for semi-complete digraphs (\cref{thm:main1}).
In \cref{sec:minfas}, we show that \textsc{Minimum Feedback Arc Set} on semi-complete digraphs is slicewise polynomial parameterized by the cycle rank of an input digraph (\cref{thm:MFAS}).
In \cref{sec:conclusion}, we mention several concluding remarks.

\section{Preliminaries}\label{sec:prelim}

For two integers $n_1$ and $n_2$, let $[n_1,n_2]$ denote the set of all integers $i$ with $n_1\le i\le n_2$.
For a positive integer $n$, let $[n]\coloneqq [1,n]$.

In this paper, all graphs and digraphs are finite and simple, but digraphs may have bi-directed edges.
For a digraph $G$, we denote the vertex set and the edge set of $G$ by $V(G)$ and $E(G)$, respectively.
For a set $A$ of vertices in a digraph $G$, we denote by $G[A]$ the subdigraph of $G$ induced by $A$, and we let $G - A \coloneqq G[V(G) \setminus A]$.
When $A=\{v\}$, we write $G-v\coloneqq G-\{v\}$.
A digraph is \emph{acyclic} if it has no directed cycles.
A digraph $G$ is \emph{strongly connected} if for every pair $(u, v)$ of vertices of $G$, there is a directed path from $u$ to $v$ in $G$.
A \emph{strongly connected component} of a digraph is a maximal strongly connected subdigraph of it.

A \emph{rooted tree} $(T,r)$ is a tree $T$ with a distinguished vertex $r\in V(T)$ called the \emph{root}, and a \emph{rooted forest} is a disjoint union of rooted trees.
Given a vertex $x$ in a rooted tree $(T,r)$, the \emph{depth} of $x$ is the length of the $(r,x)$-path in $T$, and the \emph{depth} of $(T,r)$ is the maximum depth over all vertices of $T$.
For $x,y\in V(T)$, we say that $x$ is an \emph{ancestor} of $y$, and $y$ is a \emph{descendant} of $x$, if $x$ lies on the $(r,y)$-path in $T$.
Furthermore, if $x\neq y$, then we say that $x$ is a \emph{strict ancestor} of $y$, and $y$ is a \emph{strict descendant} of $x$.
We say that $x$ is the \emph{parent} of $y$, and $y$ is a \emph{child} of $x$, if $x$ is an ancestor of $y$ that is adjacent to $y$.
Given a rooted tree $(T, r)$ and a vertex $t \in V(T)$, we denote by $T_t$ the subtree of $T$ rooted at $t$ and induced by the descendants of $t$.
We define the \emph{depth} of a rooted forest $F$ as the maximum depth over its connected components, and denote by $F_t$ the subtree rooted at $t \in V(F)$.

\subsection{Cycle rank}\label{subsection:cyclerank}
The \emph{cycle rank} of a digraph $G$, denoted by $\crank(G)$, is defined by the following recursive formula:
\begin{itemize}
    \item If $G$ is acyclic, then $\crank(G) = 0$.
    \item If $G$ is not strongly connected, then $\crank(G)$ equals the maximum cycle rank among the strongly connected components of $G$.
    \item If $G$ is strongly connected and $\lvert V(G) \rvert \geq 2$, then $\crank(G) = 1 + \min_{v \in V(G)} \crank(G - v)$.
\end{itemize}

The recursive definition of cycle rank naturally motivates the following decomposition scheme for digraphs.
A \emph{cycle rank decomposition} of a digraph $G$ is a rooted forest $F$ with $V(F) = V(G)$ that satisfies the following:
\begin{itemize}
    \item For each connected component $C$ of $F$, $G[V(C)]$ is a strongly connected component of $G$.
    \item For each $t \in V(F)$ and its child $s$ in $F$, $G[V(F_s)]$ is a strongly connected component of $G[V(F_t)]-t$.
\end{itemize}

\begin{proposition}[McNaughton~\cite{MCNAUGHTON1969CC}]\label{cr_decomposition}
    Let $G$ be a digraph.
    The cycle rank of $G$ is equal to the minimum depth of a cycle rank decomposition of $G$.
\end{proposition}

\subsection{Directed clique-width}

Let $k$ be a positive integer.
A \emph{$k$-labeled digraph} is a tuple $(G,\lab_G)$ where $G$ is a digraph, called the \emph{underlying digraph}, and $\lab_G:V(G)\to [k]$ is a function called a \emph{labeling function}.
For an integer $k \geq 1$ and distinct $i, j \in [k]$, we define the following operations.
\begin{itemize}
    \item Let $\alpha_{i,j}(G,\lab_G)$ be the $k$-labeled digraph obtained from $(G,\lab_G)$ by adding an edge from every vertex of label $i$ to every vertex of label $j$.
    \item Let $\rho_{i\to j}(G,\lab_G)$ be the $k$-labeled digraph obtained from $(G,\lab_G)$ by relabeling every vertex of label $i$ to label $j$.
    \item Let $(G,\lab_G)\oplus(H,\lab_H)$ be the disjoint union of two $k$-labeled digraphs $(G,\lab_G)$ and $(H,\lab_H)$.
\end{itemize}
By using these operations, we define the class $\DCW_k$ of $k$-labeled digraphs as follows:
\begin{itemize}
    \item The digraph $i(v)$ with a single vertex $v$ labeled by $i \in [k]$ is in $\DCW_k$.
    \item If $(G,\lab_G)$ and $(H,\lab_H)$ are two vertex-disjoint $k$-labeled digraphs in $\DCW_k$, then $(G,\lab_G)\oplus (H,\lab_H)$ is in $\DCW_k$.
    \item If $(G,\lab_G)\in \DCW_k$, then for every distinct $i, j \in [k]$, both $\alpha_{i, j}(G, \lab_G)$ and $\rho_{i \to j}(G, \lab_G)$ are in $\DCW_k$.
\end{itemize}
A \emph{directed $k$-expression} is an expression obtained by recursively applying the above operations.
The \emph{directed clique-width} of $G$, denoted by $\dcw(G)$, is the smallest integer $k$ such that $G$ is the underlying digraph of a member of $\DCW_k$.

A directed $k$-expression is naturally associated with a parse tree, called a \emph{directed $k$-expression tree}. 
A directed $k$-expression tree of a digraph $G$ is a rooted labeled tree $(T,r,\op,\mathcal{G})$, where
\begin{itemize}
    \item $(T,r)$ is a rooted tree, 
    \item if $t \in V(T)$ is a leaf, then $\op(t)=i(v)$ for some $i\in [k]$ and $v\in V(G)$, and otherwise, $\op(t)$ is one of $\oplus$, $\alpha_{i,j}$, $\rho_{i\to j}$ for distinct $i,j\in [k]$, 
    \item if $t$ is a leaf and $\op(t)=i(v_t)$ for some $i\in [k]$ and $v_t\in V(G)$, then $\mathcal{G}(t)=i(v_t)$, and the map $t\mapsto v_t$ is a bijection from the leaves of $T$ to $V(G)$,
    \item if $t$ is an internal node of $T$ and $\op(t) \in \{\alpha_{i,j}, \rho_{i\to j}\}$, then $t$ has a unique child $s$ and $\mathcal{G}(t)=\op(t) (\mathcal{G}(s))$,
    \item if $t$ is an internal node of $T$ and $\op(t)=\oplus$, then $t$ has two children $s_1, s_2$, and $\mathcal{G}(t)=\mathcal{G}(s_1)\oplus \mathcal{G}(s_2)$, and
    \item the underlying digraph of $\mathcal{G}(r)$ is $G$.
\end{itemize}
Clearly, a digraph $G$ has directed clique-width at most $k$ if and only if $G$ admits a directed $k$-expression tree.

Although an $n$-vertex digraph may admit a directed $k$-expression, its corresponding directed $k$-expression tree can have too many nodes.
However, as mentioned in the proof of~\cite[Theorem~3]{gurski21acycliccoloring}, we may assume that the directed $k$-expression is in the \emph{normal form} so that for every subexpression, each disjoint union operation is followed by edge-insertion operations and then by relabeling operations.
In this normal form, the corresponding directed $k$-expression tree has $\mathcal{O}(k^2 \cdot n)$ nodes.

\section{CW-ranking number}\label{sec:CW-ranking}

In this section, we prove that the cycle rank can be described in terms of vertex coloring with a certain connectivity property.
Recall that, for a digraph $G$ and an integer $k \geq 0$, a function $h:V(G)\to [0, k]$ is a \emph{CW-ranking of $G$ with $k$ ranks} if every closed walk $C$ of $G$ has a unique vertex $v$ for which $h(v)=\max_{u\in V(C)} h(u)$.
The \emph{CW-ranking number} of a digraph $G$, denoted by $\chi_\cw(G)$, is the minimum $k$ such that $G$ admits a CW-ranking with $k$ ranks.

The following statements are immediate from the definition of the CW-ranking number:
\begin{observation}\label{CW_ranking_property}
    Let $G$ be a digraph.
    \begin{enumerate}[label=(\alph*)]
        \item The CW-ranking number of $G$ is $0$ if and only if $G$ is acyclic.
        \item\label{CW_ranking_property-2} The CW-ranking number of $G$ equals the maximum of the CW-ranking number of the strongly connected components of $G$.
    \end{enumerate}
\end{observation}

We now establish that the two parameters are equivalent, and provide a polynomial time algorithm to compute a cycle rank decomposition from a CW-ranking.
\begin{proposition}\label{CW_ranking_equivalence}
For every digraph $G$, it holds that $\crank(G) = \chi_\cw(G)$.
Moreover, given a CW-ranking $h:V(G)\to [0,k]$, one can compute a cycle rank decomposition of $G$ of depth at most $k$ in time $\mathcal{O}((k+1)(n+m))$, where $n\coloneqq \abs{V(G)}$ and $m\coloneqq \abs{E(G)}$.
\end{proposition}
\begin{proof}
Let $k \geq 0$ be an integer and let $G$ be a digraph.
By~\cref{CW_ranking_property} and the recursive definition of the cycle rank, it suffices to prove the statement when $G$ is strongly connected.

To show that $\chi_\cw(G)\le \crank(G)$, suppose that $\crank(G) \leq k$.
By~\Cref{cr_decomposition}, $G$ has a cycle rank decomposition $(T, r)$ of depth at most~$k$.
For each $v\in V(G)$, let $\depth(v)$ be the depth of $v$ in $T$ and let $h : V(G) \to [0, k]$ be the function defined by $h(v) = k - \depth(v)$.
We claim that $h$ is a CW-ranking of $G$.
Let $C$ be a closed walk in $G$ and let $u \in V(C)$ be a vertex with minimum depth.
We claim that all vertices in $C$ belong to $T_u$.
The claim is straightforward when $u=r$, so suppose this is not the case.
Let $Z$ be the set of all proper ancestors of $u$ in $T$.
By the definition of the cycle rank decomposition, the vertices in the subtree $T_u$ of $T$ induce a strongly connected component of $G-Z$.
Since $C$ is a closed walk and $u \in V(C)$, we have $V(C)\subseteq V(T_u)$, as otherwise, $C$ contains a vertex in $Z$, which contradicts the minimality of $u$.
Thus, every vertex $v\in V(C)\setminus \{u\}$ is a proper descendant of $u$ in $T$, so $h(u) > h(v)$ for all $v\in V(C)\setminus \{u\}$.
This shows that $u$ is the unique vertex in $C$ which maximizes $h$, implying that $h$ is a CW-ranking of $G$.
Hence, we have $\chi_\cw(G) \leq k$.

To show the converse and algorithmic statement, suppose that $G$ admits a CW-ranking $h: V(G)\to [0,k]$.
We prove the existence of a cycle rank decomposition of depth at most $k$, constructively.
When $\abs{V(G)}=1$, we make a single-vertex rooted tree $(T,r)$ with $V(T)=V(G)$, which takes time $\mathcal{O}(1)$.
Assume $\abs{V(G)}\ge 2$.
Let $M\coloneqq \max_{v\in V(G)}h(v)$.
Since $G$ is strongly connected, it contains a spanning directed closed walk.
The definition of a CW-ranking applied to such a walk then guarantees the existence of a unique vertex $r \in V(G)$ satisfying $h(r)=M$.
Moreover, since $\abs{V(G)}\ge 2$ and $G$ is strongly connected, $G$ contains at least one directed cycle.
This, in particular, implies that $M \geq 1$.

Let $G_1,\ldots, G_t$ be the strongly connected components of $G-r$.
One can find these strongly connected components in time $\mathcal{O}(n+m)$ by using Tarjan's algorithm~\cite{Tarjan1972}.
For each $i \in [t]$, the restriction $h_i \coloneqq h\rvert_{V(G_i)}$ is a CW-ranking of $G_i$.
Since $r$ is the unique vertex of rank $M$ in $G$, the image of $h_i$ is contained in $[0,M-1]\subseteq [0,k-1]$.
We recursively apply this procedure to each $G_i$ to obtain a cycle rank decomposition $(T_i,r_i)$ of depth at most $k-1$.
Let $T$ be the rooted tree obtained from the disjoint union of $T_1,\ldots, T_t$ and the vertex $r$ by adding an edge from $r$ to each $r_i$, with $r$ as the root.
By construction, $(T,r)$ is a cycle rank decomposition of $G$ of depth at most~$k$.

Finally, we analyze the running time of this construction.
In each recursive step, finding the unique vertex having maximum rank and computing the strongly connected components takes time $\mathcal{O}(n+m)$.
Since the maximum rank strictly decreases at each level of the recursion, the depth of this recursion is at most $k+1$.
Therefore, the total running time is $\mathcal{O}((k+1)\cdot (n+m))$.
This proves~\Cref{CW_ranking_equivalence}.
\end{proof}

\section{An FPT algorithm for bounded directed clique-width classes}\label{sec:bdddcw}

In this section, we prove~\Cref{thm:main}, which we restate for the reader's convenience.

\mainthm*

\begin{proof}
Let $G$ be an $n$-vertex digraph and let $\mathbf{T} = (T,r,\op, \mathcal{G})$ be a directed $k$-expression tree of $G$.
In particular, we may assume that the directed $k$-expression of $G$ is given in the normal form so that~$\lvert V(T) \rvert = \mathcal{O}(k^2 \cdot n)$.
For a node $t$ of $T$, let $(G_t,\lab_t)\coloneqq \mathcal{G}(t)$ be the $k$-labeled digraph associated with $t$.

Let $f$ be a function from $V(G_t)$ to $[0,w]$.
For such a function $f$ and a vertex $v\in V(G_t)$, we say that $f(v)$ is the \emph{rank} of $v$ with respect to $f$.
We say that $f$ is \emph{valid} if it is a CW-ranking of $G_t$ with $w$ ranks.
For each $\ell\in [k]$, we define
\[V(t,\ell)\coloneqq \left\{v\in V(G_t):\lab_t(v)=\ell\right\}.\]
For a vertex $v\in V(G_t)$ with rank $f(v)=m$, we define its \emph{in-label set} $I(t,v,f)\subseteq [k]$ and \emph{out-label set} $O(t,v,f)\subseteq [k]$ as follows:
\begin{itemize}
    \item $\ell\in I(t,v,f)$ if and only if there is a directed path in $G_t$ from $V(t,\ell)$ to $v$ all of whose vertices have rank at most $m$, and
    \item $\ell\in O(t,v,f)$ if and only if there is a directed path in $G_t$ from $v$ to $V(t,\ell)$ all of whose vertices have rank at most $m$.
\end{itemize}
We note that the maximum rank on any such path is $m$, since $f(v)=m$.

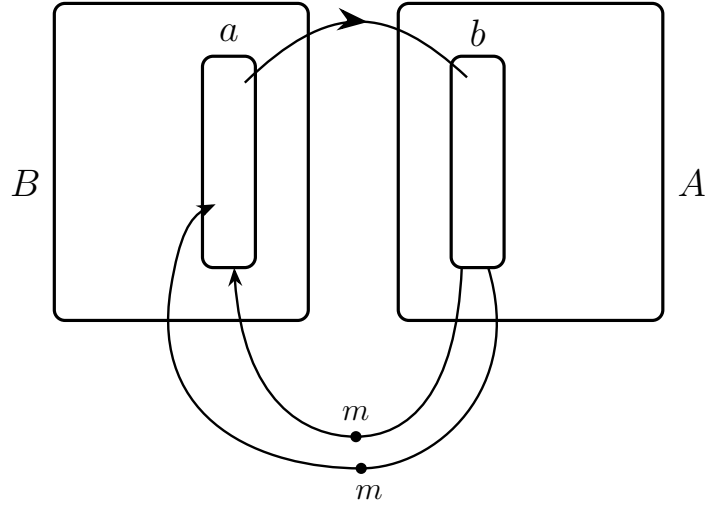
\begin{figure}[t]
\centering
    \begin{tikzpicture}[
    scale=0.7,
    line width=0.9pt,
    midarrow/.style={
        postaction={decorate},
        decoration={markings, mark=at position 0.55 with {\arrow{Stealth[length=4mm]}}}
    },
    >=Stealth
]

\draw[very thick, rounded corners] (0,1) rectangle (4.8,7);          % B
\draw[very thick, rounded corners] (6.5,1) rectangle (11.5,7);     % A

\draw[very thick, rounded corners] (2.8,2) rectangle (3.8,6);      % a
\draw[very thick, rounded corners] (7.5,2)   rectangle (8.5,6);        % b

\node at (-0.55,3.6) {\Large $B$};
\node at (12.05,3.6) {\Large $A$};
\node at (3.3,6.45)  {\Large $a$};
\node at (8,6.45)  {\Large $b$};

\draw[midarrow] (3.6,5.5)
    .. controls (5.2,7.05) and (6.3,7.0) .. (7.8,5.6);

\draw[->] (7.7,2)
    .. controls (7.6,-0.4) and (6.7,-1.2) .. (5.7,-1.2)   
    .. controls (4.7,-1.2) and (3.5,-0.5) .. (3.4,2);

\draw[->] (8.2,2)
    .. controls (8.9,-0.3) and (7.2,-1.8) .. (5.8,-1.8)   
    .. controls (3.9,-1.8) and (1.8,-0.9) .. (2.2,1.6)
    .. controls (2.4,2.85) and (2.65,3.05) .. (3.05,3.2);

\fill (5.7,-1.2) circle (3pt);
\fill (5.8,-1.8) circle (3pt);
\node at (5.7,-0.75) {\large $m$};
\node at (5.95,-2.25) {\large $m$};

\end{tikzpicture}
\caption{When adding all edges from vertices of label $a$ to vertices of label $b$, we need to determine whether a new closed walk containing two vertices of maximum rank is created. 
To this end, we maintain the information $M[A,B;m]$, which records the number of vertices of rank $m$ on a path from a vertex with out-label set $B$ to a vertex with in-label set $A$, capped at $2$.
As illustrated in the figure, if there are two such vertices with paths $P$ and $Q$ (not necessarily disjoint), then $P\cup Q$ together with two new edges forms a closed walk having two vertices of maximum value, so it would not be a CW-ranking.}
\label{fig:twomaximum}
\end{figure}

We define a function $M=M_f$ which depends on $f$.
A function $M:2^{[k]}\times 2^{[k]}\times [0,w]\to [0,2]$ is the \emph{multiplicity-profile} of $G_t$ with respect to $f$ if for every $A,B\subseteq [k]$ and every $m\in [0,w]$, 
\[M[A,B;m]=\min\left(2,\abs{\{v\in V(G_t):\text{$f(v)=m$, $I(t,v,f)=A$, and $O(t,v,f)=B$}\}}\right).\]

\Cref{fig:twomaximum} illustrates why it is necessary to store this information. 
When we add edges from vertices of label $a$ to vertices of label $b$, we need to determine whether a new closed walk containing two vertices of maximum rank is created; if such a closed walk is created, then the current function will not be extended to a valid function. 
Thus, we need to store the information about whether there are two vertices where each vertex is a vertex of maximum rank $m$ on some path from a vertex of label $b$ to a vertex of label $a$.

We also explain why it is necessary to consider the in-label set and out-label set of a vertex, rather than recording the labels of the endpoints of a directed path passing through it. See \Cref{fig:keepingendpoints} for an example. In the figure, there is a directed path $P$ from a vertex of label $a$ to a vertex of label $b$, and a directed path $Q$ from a vertex of label $a$ to a vertex of label $c$, where $v$ is a common vertex of $P$ and $Q$ and is the unique vertex of maximum value on each path. Suppose that the label $c$ is relabeled to $b$. Then $P$ and $Q$ become two distinct paths from a vertex of label $a$ to a vertex of label $b$, both having $v$ as their unique maximum-value vertex. Then adding edges from vertices of label $b$ to vertices of label $a$ does not create a closed walk whose maximum value is attained by at least two vertices. To correctly capture such situations, it is natural to keep track of all labels $x$ for which there exists a path from $v$ to a vertex of label $x$. To correctly capture this information, we record the entire out-label set of $v$. By symmetry, we also maintain the in-label set.

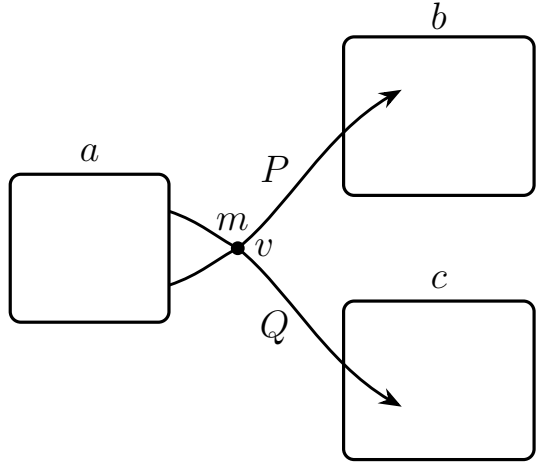
\begin{figure}
\centering
\begin{tikzpicture}[
    scale=0.7,
    line width=0.9pt,
    flow/.style={line width=1.1pt, -{Stealth[length=3mm]}},
    node/.style={fill=black, line width=1pt}
]

  \draw[very thick, rounded corners] (0,2.6)  rectangle (3.0,5.4);    
  \draw[very thick, rounded corners] (6.3,5) rectangle (9.9,8);   
  \draw[very thick, rounded corners] (6.3,0) rectangle (9.9,3);  
  
  \node at (1.5,5.8) {\Large $a$};
  \node at (8.1,8.4) {\Large $b$};
  \node at (8.1,3.4) {\Large $c$};

  \draw[flow] (3.0,4.7)
      .. controls (3.6,4.5) and (3.9,4.2) .. (4.3,4.0)
      .. controls (5.3,4.8) and (6.0,6.3) .. (7.4,7.0);
  \draw[flow] (3.0,3.3)
      .. controls (3.6,3.5) and (3.9,3.8) .. (4.3,4.0)
      .. controls (5.3,3.2) and (6.0,1.7) .. (7.4,1.0);

  \draw[node] (4.3,4.0) circle (3pt);
  \node at (4.2,4.5) {\Large $m$};
  \node at (4.8,4) {\Large $v$}; 
  \node at (5,5.5) {\Large $P$}; 
  \node at (5,2.5) {\Large $Q$}; \end{tikzpicture}
\caption{An example describing why we need to store the in-set and out-set of a vertex.}
\label{fig:keepingendpoints}
\end{figure}

Let $\mathcal{M}_t$ be the collection of all functions $M$ that are multiplicity profiles of $G_t$ with respect to some valid function. 
It is straightforward from the definition of $\mathcal{M}_t$ that, for the root $r$ of the directed $k$-expression tree of $G$, the input $(G, \mathbf{T}, w)$ is a \textsc{Yes}-instance if and only if $\mathcal{M}_r \neq \varnothing$.
Thus, it suffices to recursively compute $\mathcal{M}_t$.

\medskip
\noindent
\textbf{Algorithm.}
We compute $\mathcal{M}_t$ using bottom-up dynamic programming. 
In the algorithm, we compute the set $\mathcal{R}_t$ for each node $t$, and in the correctness proof, we show that $\mathcal{R}_t=\mathcal{M}_t$. 
For each case, we take $t$ as the current node.

\smallskip
\noindent
\textbf{Case 1.} $t$ is a leaf in $T$. 

\noindent
Let $\op(t) = i(v)$ for some $i \in [k]$ and $v \in V(G)$.
We branch on all possible ranks $z \in [0,w]$ of $v$.
Since $\mathcal{G}(t)$ consists of a single vertex $v$ with label $i$ and rank $z$, the only path from and to $v$ is the trivial path of length $0$.
Thus, for each $z$, we construct a profile $M:2^{[k]}\times 2^{[k]}\times [0,w]\to [0,2]$ by setting, for every $A, B \subseteq [k]$ and $m \in [0, w]$,
\[
    M[A,B;m]=
    \begin{cases}
        1, &\text{if $A=B=\{i\}$ and $m=z$,}\\
        0, &\text{otherwise,}
    \end{cases}
\]
and add $M$ to $\mathcal{R}_t$.

\smallskip
\noindent
\textbf{Case 2.} $\op(t)=\oplus$.

\noindent
Let $t_1$ and $t_2$ be the two children of $t$.
We iterate over all pairs of records $M_1\in \mathcal{R}_{t_1}$ and $M_2\in \mathcal{R}_{t_2}$.
Since $G_t$ is the disjoint union of $G_{t_1}$ and $G_{t_2}$, every directed path of $G_t$ is contained in $G_{t_1}$ or $G_{t_2}$.
Furthermore, for each vertex $v\in V(G_t)$, its in-label set and out-label set in $G_t$ coincide with those in whichever of $G_{t_1}$ and $G_{t_2}$ contains $v$. 
Thus, for every $A, B \subseteq [k]$ and $m \in [0, w]$, the number of vertices in $G_t$ of rank $m$ with in-label set $A$ and out-label set $B$ is equal to the sum of the corresponding counts in $G_{t_1}$ and $G_{t_2}$. 
We therefore construct a new profile $M$ as the pointwise sum of the two profiles, capped at $2$.
Specifically, for each $A,B\subseteq [k]$ and $m\in [0,w]$, we let
\[M[A,B;m] \coloneqq \min\left(2, M_1[A,B;m]+M_2[A,B;m]\right)\]
and add $M$ to $\mathcal{R}_t$.

\smallskip
\noindent
\textbf{Case 3.} $\op(t)=\rho_{a\to b}$ for distinct $a,b\in [k]$.

\noindent
Let $t'$ be the unique child of $t$.
For each $S\subseteq [k]$, we define its \emph{relabeled set} $\rho_{a\to b}(S)$ as $(S\setminus\{a\})\cup\{b\}$ if $a\in S$, and $\rho_{a\to b}(S)=S$ otherwise.
We iterate over each record $M'\in\mathcal{R}_{t'}$.
Since relabeling does not add new edges, the new in-label and out-label sets are updated by applying the operation $\rho_{a\to b}$.
Thus, we construct a new profile $M$ by setting, for each  $A,B\subseteq [k]$ and $m\in [0,w]$,
\[
    M[A,B;m] \coloneqq \min\left(2, \sum_{\substack{A',B'\subseteq [k] \\ \rho_{a\to b}(A')=A \\ \rho_{a\to b}(B')=B}}M'[A',B';m]\right),\]
and add $M$ to $\mathcal{R}_t$.

\smallskip
\noindent
\textbf{Case 4.} $\op(t)=\alpha_{a,b}$ for distinct $a,b\in [k]$.

\noindent
Let $t'$ be the unique child of $t$.
We iterate over all records $M'\in\mathcal{R}_{t'}$.
Before we construct the new profile $M$, we first check whether, by applying the operation $\alpha_{a,b}$, a new closed walk with a non-unique maximum rank is created.
Observe that such a closed walk is formed if and only if there are at least two vertices of the same maximum rank, each lying on a directed path in $G_{t'}$ from a vertex of label $b$ to a vertex of label $a$.
Specifically, for each $m \in [0,w]$, we compute
\[
    \mathsf{C}(m) \coloneqq \sum_{\substack{A,B\subseteq [k]\\ b\in A,\ a\in B}}M'[A,B;m]
\]
and discard $M'$ if $\mathsf{C}(m) \geq 2$ for some $m$.

When $M'$ is not discarded, we construct a new profile to reflect the newly created directed walks.
For each $m\in [0,w]$, we define two sets $R^\downarrow_a(m)$ and $R^{\uparrow}_b(m)$ as follows:
\begin{align*}
    R^\downarrow_a(m)&=\bigcup\left\{L\subseteq [k]:\text{there are $A\subseteq [k]$ and $z\le m$ such that $M'[L,A;z]\ge 1$ and $a\in A$}\right\};\\
    R^\uparrow_b(m)&=\bigcup\left\{L\subseteq [k]:\text{there are $B\subseteq [k]$ and $z\le m$ such that $M'[B,L;z]\ge 1$ and $b\in B$}\right\}.
\end{align*}
The set $R^\downarrow_a(m)$ is the union of label sets which can reach label $a$ by using paths of maximum rank at most $m$, and the set $R^\uparrow_b(m)$ is the union of label sets which can be reached from label $b$ by using paths of maximum rank at most $m$.
For each $A',B'\subseteq [k]$ and $m\in [0,w]$, we define $\mathsf{E}(A',B',m)$ as the pair $(A^*,B^*)$ of subsets of $[k]$, where
\[A^*=\begin{cases}
    A'\cup R^\downarrow_a(m),&\text{if $b\in A'$,}\\
    A', &\text{otherwise,}
\end{cases}
\quad\text{and}\quad
B^*=\begin{cases}
    B'\cup R^\uparrow_b(m),&\text{if $a\in B'$,}\\
    B', &\text{otherwise.}
\end{cases}\]
We then construct a new profile $M$ by setting, for every $A,B\subseteq [k]$ and $m\in [0,w]$,
\[M[A,B;m]=\min\left(2, \sum_{\substack{A',B'\subseteq [k]\\ \mathsf{E}(A',B',m)=(A,B)}}M'[A',B';m]\right),\]
and add $M$ to $\mathcal{R}_t$.
We refer to~\Cref{fig:Example1} for an illustration.
\begin{figure}[t]
    \centering
    \includegraphics[scale=1.0]{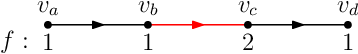}
    \caption{For each vertex $v_i$, let $\lab(v_i)=i$. The digraph is obtained by applying the operation $\alpha_{b,c}$ to the base digraph consisting of the black edges. (The added edge $(v_b,v_c)$ is drawn in red.) Let $f$ be the guessed ranking function. 
    Then for $m=2$, we have $M'[\{c\},\{c,d\};2]=1$ and all the other values are $0$. 
    We compute the reachability sets $R^{\downarrow}_b(2)=\{a,b\}$ and $R^{\uparrow}_c(2)=\{c,d\}$. 
    Thus, for $A'=\{c\}$ and $B'=\{c,d\}$, we have $\mathsf{E}(A',B',2)=(A^*,B^*)$, where $A^*=A'\cup R^{\downarrow}_b=\{a,b,c\}$ and $B^*=B'=\{c,d\}$, since $c\in A'$ and $b\notin B'$. 
    Therefore, $M[\{a,b,c\},\{c,d\};2]=1$ and $M[A,B;2]=0$ for the other subsets $A,B$.}
    \label{fig:Example1}
\end{figure}

\medskip
\noindent
\textbf{Correctness.} 
We prove that $\mathcal{R}_t = \mathcal{M}_t$ for every $t$ by bottom-up induction on the directed $k$-expression tree $\mathbf{T}$ of $G$.

For each node $t$ in $T$ and a valid function $h : V(G_t) \to [0, w]$, let us denote by $M_{t, h}$ the multiplicity-profile of $G_t$ with respect to $h$.

\smallskip
\noindent
\textbf{Base case.} $t$ is a leaf in $T$.

\noindent
Let $\op(t) = i(v)$ for some $i \in [k]$ and $v \in V(G)$.
To show that $\mathcal{R}_t \subseteq \mathcal{M}_t$, take $M \in \mathcal{R}_t$.
Let $z \in [0, w]$ be the guessed rank of $v$ when constructing $M$, and let $h : V(G_t) \to [0, w]$ be the function defined by $h(v) = z$.
Then $h$ is valid since $G_t$ has no edge.
Furthermore, since the only path in $G_t$ is the trivial path at $v$, its in-label set and out-label set are $\{i\}$.
This shows that $M_{t, h} = M$ and so $M \in \mathcal{M}_t$.
Conversely, if $h^*$ is a valid CW-ranking of $G_t$ with $h^*(v)=z'$, then the algorithm generates the corresponding profile $M$ for the rank $z'$ and the label $i$, which is exactly $M_{t, h^*}$. 
Therefore, it holds that $\mathcal{R}_t = \mathcal{M}_t$.

\smallskip
\noindent
\textbf{Disjoint union case.} $\op(t)=\oplus$. 

\noindent
Let $t_1$ and $t_2$ be the two children of $t$ in $T$.
Then $G_t$ is the disjoint union of $G_{t_1}$ and $G_{t_2}$.
We first show that $\mathcal{R}_t \subseteq \mathcal{M}_t$.
Take $M \in \mathcal{R}_t$.
Then $M$ is obtained by the pointwise summation of some $M_1 \in \mathcal{R}_{t_1}$ and $M_2 \in \mathcal{R}_{t_2}$.
By the induction hypothesis, for each $i \in [2]$, there is a valid function $h_i : V(G_{t_i}) \to [0, w]$ such that $M_i = M_{t_i, h_i}$.
Let $h :V(G_t) \to [0, w]$ be a function such that $h \rvert_{V(G_{t_i})} = h_i$ for each $i$.

We claim that $M = M_{t, h}$.
Since there are no edges between $G_{t_1}$ and $G_{t_2}$, the function $h$ of $G_t$ is valid, and for each $v\in V(G_t)$, its in-label set and out-label set in $G_t$ match those in the corresponding subdigraphs obtained from $t_1$ and $t_2$.
Since the subdigraphs are disjoint, the total number of vertices with rank $m$, in-label set $A$, and out-label set $B$ in $G_t$ is the sum of the number of such vertices in $G_{t_1}$ and $G_{t_2}$, capped at $2$.
Thus, $M_{t, h}$ is the pointwise addition of $M_1$ and $M_2$ capped at $2$, which matches~$M$.
This shows that $\mathcal{R}_t \subseteq \mathcal{M}_t$.

To show the converse, take $M_{t, h^*} \in \mathcal{M}_t$ for some valid function $h^* : V(G_t) \to [0, w]$.
Then its restrictions $h^*_1=h^*\vert_{V(G_{t_1})}$ and $h^*_2 = h^*\vert_{V(G_{t_2})}$ are valid on $G_{t_1}$ and $G_{t_2}$, respectively.
By the induction hypothesis, we have $M_{t_1, h_1^*} \in \mathcal{R}_{t_1}$ and $M_{t_2, h_2^*} \in \mathcal{R}_{t_2}$, so the algorithm includes the  pointwise sum of $M_{t_1, h_1^*}$ and $M_{t_2, h_2^*}$ in $\mathcal{R}_t$.
Since every directed path in $G_t$ is contained in one of $G_{t_1}$ and $G_{t_2}$, for every $i \in [2]$ and $v \in V(G_{t_i})$, it holds that 
\[
    I(t,v,h^*)=I(t_i,v,h_i^*)
    \quad\text{and}\quad
    O(t,v,h^*)=O(t_i,v,h_i^*).
\]
This shows that the pointwise sum of $M_{t_1, h_1^*}$ and $M_{t_2, h_2^*}$ is equal to $M_{t,h^*}$, and so $M_{t, h^*} \in \mathcal{R}_t$.
Therefore, we conclude that $\mathcal{R}_t = \mathcal{M}_t$. 

\smallskip
\noindent
\textbf{Relabeling case.} $\op(t)=\rho_{a\to b}$ for distinct $a, b \in [k]$.

\noindent
Let $t'$ be the child of $t$ in $T$.
Then $\mathcal{G}(t) = \rho_{a \to b} (\mathcal{G}(t'))$ and, in particular, we have $G_t = G_{t'}$.
First, to show that $\mathcal{R}_t \subseteq \mathcal{M}_t$, take $M \in \mathcal{R}_t$.
Then there is $M' \in \mathcal{R}_{t'}$ such that $M$ is obtained from $M'$ by applying the procedure in Case~3.
By the induction hypothesis, there is a valid function $h : V(G_{t'}) \to [0, w]$ such that $M' = M_{t', h}$.
Observe that, since $E(G_t) = E(G_{t'})$, the function $h$ is also valid for $G_t$.

We claim that $M = M_{t, h}$.
Since the operation $\rho_{a \to b}$ only changes the label of some vertices in $G_{t'}$ and does not change the edge set, we have
\[
    I(t,v,h)=I(t',v,h)
    \quad \text{and} \quad
    O(t,v,h)=O(t',v,h)
\]
for each $v \in V(G_t)$.
Thus, for every $A, B \subseteq [k]$ and $m \in [0, w]$, the vertices of rank $m$ with in-label set $A$ and out-label set $B$ in $G_t$ are precisely the vertices of rank $m$ whose in-label set $A'$ and out-label set $B'$ in $G_{t'}$ satisfy $\rho_{a \to b}(A') = A$ and $\rho_{a \to b}(B') = B$.
This shows that
\[
    M_{t, h}[A, B;m] = \min\left(2, \sum_{\substack{A', B' \subseteq [k] \\ \rho_{a \to b}(A') = A \\ \rho_{a \to b} (B') = B}}  M'[A', B';m]\right).
\]
Since $M' = M_{t', h}$, the right-hand side is exactly the value assigned to $M[A,B;m]$ by Case~3. This shows that $M=M_{t,h}$ and so
$M \in \mathcal{M}_t$, implying that $\mathcal{R}_t \subseteq \mathcal{M}_t$.

To show the converse, take $M_{t, h^*} \in \mathcal{M}_t$ for some valid function $h^* : V(G_t) \to [0, w]$.
Again, since $E(G_t) = E(G_{t'})$, the function $h^*$ is a valid function for $G_{t'}$, so the induction hypothesis shows that $M_{t', h^*} \in \mathcal{R}_{t'}$.
If we consider the profile $M^*$ obtained by applying the procedure in Case~3 to $M_{t', h^*}$, the same argument as above shows that $M^* = M_{t, h^*}$.
It therefore follows that $M_{t, h^*} \in \mathcal{R}_t$ and so $\mathcal{M}_t \subseteq \mathcal{R}_t$.

\smallskip
\noindent
\textbf{Edge-insertion case.} $\op(t)=\alpha_{a,b}$ for distinct $a, b \in [k]$. 

Let $t'$ be the child of $t$ in $T$.
Then $G_t$ is obtained from $G_{t'}$ by adding all possible edges from vertices of label $a$ to vertices of label $b$ in $\lab_{t'}$.

\begin{claim}\label{alpha_update}
A valid function $h : V(G_{t'}) \to [0, w]$ for $G_{t'}$ is valid for $G_t$ if and only if for every $m \in [0, w]$, there is at most one vertex $v$ such that 
\[
    h(v) = m, \quad 
    b \in I(t', v, h), \quad 
    \text{and} \quad 
    a \in O(t', v, h).
\]
\end{claim}
\begin{clproof}
We prove the contrapositive of the statement.
First, suppose that there are two distinct vertices $x, y \in V(G_t)$ of rank $m$ satisfying the desired conditions.
Let $P_x$ and $P_y$ be directed paths in $G_{t'}$ from a vertex of label $b$ to a vertex of label $a$ that contain $x$ and $y$, respectively.
Since $\mathcal{G}(t) = \alpha_{a, b}(\mathcal{G}(t'))$, by concatenating $P_x$, $P_y$, and newly created edges, we obtain a directed closed walk in $G_t$ whose maximum rank is $m$ and contains both $x$ and $y$.
Thus, the function $h$ is not valid for $G_t$.

To show the converse, suppose that $h$ is not valid for $G_t$.
Let $W$ be a directed closed walk in $G_t$ that contains at least two distinct vertices of maximum rank under $h$, and let $m$ be the maximum rank of the vertices in $W$.
Since $h$ is valid for $G_{t'}$, the walk $W$ contains at least one newly created edge in $G_t$.
Now, let $P_1, \ldots, P_s$ be the segments of $W$ obtained by removing the edges in $E(G_t) \setminus E(G_{t'})$.
Then each $P_i$ is a directed path in $G_{t'}$ from a vertex of label $b$ to a vertex of label $a$, and all vertices in $P_i$ are of rank at most $m$.

Let $x \in V(W)$ be a vertex of rank $m$, and let $P_i$ be the segment of $W$ that contains $x$.
Since $P_i$ starts at a vertex of label $b$, the subpath of $P_i$ from the first vertex to $x$ shows that $b \in I(t', x, h)$.
Similarly, since $P_i$ ends at a vertex of label $a$, the subpath of $P_i$ from $x$ to the last vertex shows that $a \in O(t', x, h)$.
Thus, every vertex in $W$ of maximum rank satisfies the desired conditions.
Since $W$ contains at least two vertices of rank $m$, it therefore follows that there are at least two vertices satisfying the desired conditions.
This proves~\Cref{alpha_update}.
\end{clproof}

We first show that $\mathcal{R}_t \subseteq \mathcal{M}_t$.
Take $M \in \mathcal{R}_t$ and let $M' \in \mathcal{R}_{t'}$ be the record to which we applied the procedure in Case~4 to obtain $M$.
By the induction hypothesis, there is a valid function $h : V(G_{t'}) \to [0, w]$ such that $M' = M_{t', h}$.
Since $M'$ is not discarded at the first step, we have
\[
    \sum_{\substack{A,B\subseteq [k]\\ b\in A,\ a\in B}}M'[A,B;m] \leq 1
\]
for every $m \in [0, w]$.
This means that for every $m$, there is at most one vertex $v$ satisfying $h(v) = m$, $b \in I(t', v, h)$, and $a \in O(t', v, h)$.
Thus,~\Cref{alpha_update} shows that $h$ is valid for $G_t$.

It remains to show that $M = M_{t, h}$.
Let us begin with the following claim:

\begin{claim}\label{alpha_update2}
Let $m \in [0,w]$, let $v \in V(G_t)$ such that $h(v) = m$, and let $A' = I(t', v, h)$ and $B' = O(t', v, h)$.
It holds that $A^* = I(t, v, h)$ and $B^* = O(t, v, h)$, where $(A^*, B^*) = \mathsf{E}(A', B', m)$.
\end{claim}
\begin{clproof}
Since the proof for each equality is symmetric, we only prove that $A^* = I(t, v, h)$.
Clearly, we have $A' \subseteq I(t, v, h)$.
Suppose that there is a label $\ell \in I(t, v, h) \setminus A'$.
Then there is a directed path $P$ in $G_t$ from a vertex of label $\ell$ to $v$ whose maximum rank is $m$ and contains an edge in $E(G_t) \setminus E(G_{t'})$.
Now, consider the first and the last newly added edges, say $e$ and $e'$, respectively, in $P$.
The section before $e$ is a path in $G_{t'}$ from label $\ell$ to a vertex of label $a$, so $\ell \in R_a^\downarrow(m)$. 
Similarly, the part after $e'$ is a path in $G_{t'}$ from a vertex of label $b$ to $v$, so $b \in A'$.
Thus, every label in $I(t, v, h) \setminus A'$ belongs to $R_a^\downarrow(m)$, and such new labels occur only when $b \in A'$.

Conversely, if $b \in A'$ and $\ell \in R_a^\downarrow(m)$, then in $G_{t'}$, there is a directed path from a vertex of label $\ell$ to a vertex of label $a$ and a directed path from a vertex of label $b$ to $v$, both of which have maximum rank at most $m$.
By concatenating these paths with a newly created edge in $G_t$, we obtain a directed path in $G_t$ from a vertex of label $\ell$ to $v$ whose maximum rank is $m$.
These arguments show that
\[
    I(t, v, h) = 
    \begin{cases}
         A' \cup R_a^\downarrow(m), & b \in A', \\
         A', & \text{otherwise,}
    \end{cases}
\]
showing that $A^* = I(t, v, h)$.
This proves~\Cref{alpha_update2}.
\end{clproof}

By~\Cref{alpha_update2}, for each $A, B \subseteq [k]$ and $m \in [0, w]$, a vertex is  counted in $M_{t, h}[A, B;m]$ if and only if its in-label set $A'$ and out-label set $B'$ in $G_{t'}$ satisfy $\mathsf{E}(A', B', m) = (A, B)$.
Since $M' = M_{t', h}$, we have
\[
    M_{t, h}[A, B;m]
    = \min \left(2, \sum_{\substack{A', B' \subseteq [k] \\ \mathsf{E}(A', B', m) = (A, B)}} M'[A', B';m]\right).
\]
This shows that $M_{t, h} = M$ and so $M \in \mathcal{M}_t$.
Thus, we have $\mathcal{R}_t \subseteq \mathcal{M}_t$.

To show the converse, fix a valid function $h^* : V(G_t) \to [0, w]$ and consider $M_{t, h^*} \in \mathcal{M}_t$.
Since $G_{t'}$ is a subdigraph of $G_t$, the function $h^*$ is also valid for $G_{t'}$, so by the induction hypothesis, we have $M_{t', h^*} \in \mathcal{R}_{t'}$.
Observe that, by~\Cref{alpha_update} applied to $h^*$, the profile $M_{t', h^*}$ is not discarded at the first phase of Case~4.

Let $M^*$ be the profile obtained from $M_{t', h^*}$ by applying the procedure in Case~4.
By the same argument as in~\Cref{alpha_update2}, for each vertex $v$ of rank $m$, its in-label set and out-label set in $G_t$ are obtained from its in-label and out-label sets in $G_{t'}$ by applying $\mathsf{E}(\cdot,\cdot,m)$.
By a similar argument above, it therefore follows that $M^* = M_{t, h^*}$.
Thus, we have $M_{t, h^*} \in \mathcal{R}_t$, so $\mathcal{M}_t \subseteq \mathcal{R}_t$.

\medskip

Since we have considered all possible cases, we conclude that $\mathcal{R}_t = \mathcal{M}_t$ for every $t \in V(T)$. 
In particular, it holds that $\mathcal{R}_r = \mathcal{M}_r$ at the root $r$ of $\mathbf{T}$.
This proves the correctness of the algorithm.

\medskip
\noindent
\textbf{Time complexity.}
Let $D\coloneqq (w+1) \cdot 4^k$ and $S\coloneqq 3^D$.
Since every profile has $D$ entries, there are at most $S$ profiles at each node.
Now, we analyze the running time at each node $t$ of the tree $T$.
\begin{itemize}
    \item $t$ is a leaf: The algorithm iterates over $w+1$ possible ranks.
    For each rank, initializing a profile takes time $\mathcal{O}(D)$.
    Thus, the algorithm runs in time $\mathcal{O}(D(w+1))$.
    \item $\op(t) = \oplus$: The algorithm iterates over at most $S^2$ pairs of profiles from the children.
    For each pair, the pointwise addition takes time $\mathcal{O}(D)$.
    Thus, the algorithm runs in time $\mathcal{O}(DS^2)$.
    \item $\op(t) = \rho_{a\to b}$ for distinct $a, b \in [k]$: The algorithm iterates over at most $S$ profiles.
    For each profile, the summation takes time $\mathcal{O}(D)$.
    Thus, the algorithm runs in time $\mathcal{O}(DS)$.
    \item $\op(t) = \alpha_{a,b}$ for distinct $a, b \in [k]$: The algorithm iterates over at most $S$ profiles.
    For each profile, computing $\mathsf{C}(m)$ takes time $\mathcal{O}(4^k)$.
    Determining $R^{\downarrow}_a(m)$ and $R^\uparrow_b(m)$ takes $\mathcal{O}(2^k (w+1))$ time.
    Constructing a new profile takes time $\mathcal{O}(D)$.
    Thus, the algorithm runs in time $\mathcal{O}(DS)$.
\end{itemize}
The total running time at each node is dominated by the disjoint union operation.
Since the directed $k$-expression tree $\mathbf{T}$ has $\mathcal{O}(k^2 \cdot n)$ nodes, the total running time of the algorithm is 
\[\mathcal{O}(DS^2k^2\cdot n)=\mathcal{O}((w+1)4^k9^{(w+1)4^k}k^2\cdot n)=2^{\mathcal{O}((w+1)4^k)}n.\]
This proves the existence of an FPT algorithm that decides if the cycle rank of $G$ is at most $w$ in the desired time.

\medskip
\noindent
\textbf{Constructing the cycle rank decomposition.}
We now explain how to construct a cycle rank decomposition of $G$ from the dynamic programming table. During the computation of the sets $\mathcal{R}_t$, whenever a valid profile $M\in \mathcal{R}_t$ is generated, we store one choice of valid profiles at the child node or nodes from which $M$ is obtained. More precisely, if $\op(t)=\oplus$ and $t$ has children $t_1$ and $t_2$, then for each $M\in\mathcal{R}_t$ we store one pair $(M_1,M_2)\in\mathcal{R}_{t_1}\times\mathcal{R}_{t_2}$ that generates $M$. If $\op(t)$ is either $\rho_{a\to b}$ or $\alpha_{a,b}$ and $t'$ is the unique child of $t$, we store one profile $M'\in\mathcal{R}_{t'}$ that generates $M$.

Suppose that $(G,w)$ is a \yes-instance. Then $\mathcal{R}_r\neq\varnothing$, where $r$ is the root of $T$. We choose an arbitrary profile $M^*\in\mathcal{R}_r$ and backtrack through the stored choices, starting from $(r,M^*)$. At each internal node, we recursively proceed to the stored profile or profiles at its child nodes. In this way, we select exactly one valid profile for each node of $T$.

For each leaf $t$ of $T$ with $\op(t)=i_t(v_t)$, the selected profile was generated in Case~1 by assigning some rank $z\in[0,w]$ to $v_t$. We set $h(v_t)\coloneqq z$. Since the selected profiles correspond to a consistent sequence of valid computations of the dynamic program, the resulting function $h:V(G)\to[0,w]$ is a CW-ranking of $G$.

By~\Cref{CW_ranking_equivalence}, a cycle rank decomposition can then be constructed from the CW-ranking $h$ in time $\mathcal{O}((w+1)(n+m))$, where $m\coloneqq |E(G)|$. Hence, a cycle rank decomposition of $G$ can be constructed within the same asymptotic running time.
\end{proof}

\section{An FPT algorithm for semi-complete digraphs}\label{sec:semicomplete}

In this section, we prove~\Cref{thm:main1}, which we restate for the reader's convenience.

\mainonethm*

Let us begin with some definitions.
A \emph{directed path-decomposition} of a digraph $G$ is a sequence $(X_1, \ldots, X_m)$ of vertex subsets of $G$, which are called the \emph{bags}, satisfying the following properties:
\begin{itemize}
    \item $\bigcup_{i=1}^m X_i = V(G)$.
    \item For each $1 \leq i < j < k \leq m$, it holds that $X_i \cap X_k \subseteq X_j$.
    \item For each $(u, v) \in E(G)$, there exist $1 \leq i \leq j \leq m$ such that $u\in X_i$ and $v\in X_j$.
\end{itemize}
The \emph{width} of a directed path-decomposition $(X_1, \ldots, X_m)$ is $\max_{i \in [m]} \lvert X_i \rvert  -1$.
The \emph{directed path-width} of $G$, denoted by $\dpw(G)$, is the minimum width among all directed path-decompositions of~$G$.

\begin{proposition}[Lemma 10 in~\cite{Gruber2012}]\label{prop:dpwcr}
For every digraph $G$, it holds that $\dpw(G) \leq \crank(G)$.
\end{proposition}  

Although such a decomposition is not part of the input, we can compute one, since there is an FPT algorithm that computes a directed path-decomposition of small width:

\begin{theorem}[Theorem 4.13 in~\cite{fominpilipczuk19dpw}]\label{thm:dpwsemi}
There is an algorithm that, given a semi-complete digraph $G$ on $n$ vertices and an integer $w \geq 0$, either computes a directed path-decomposition of width at most $w$ or correctly concludes that no such path-decomposition exists in time $2^{\mathcal{O}(w \log w)} \cdot n^2$.
\end{theorem}

We will also use that, for semi-complete digraphs, a directed
path-decomposition can be converted into a directed $k$-expression in
FPT time, as stated in the following lemma:

\begin{lemma}[Lemma A.2 in~\cite{FominP2013}]\label{lem:dpwdcwsemi}
For every semi-complete digraph $G$, we have $\dcw(G)\le \dpw(G)+2$.
Moreover, given an $n$-vertex semi-complete digraph and its directed path-decomposition of width $k$, one can compute a directed $(k+2)$-expression of $G$ in time $\mathcal{O}(n^2)$.
% Furthermore, the directed $(k+2)$-expression tree of such a directed $(k+2)$-expression has $\mathcal{O}(k)\cdot n$ nodes.
\end{lemma}

We are now ready to prove our main theorem.

\begin{proof}[Proof of~\Cref{thm:main1}]
Let $G$ be a semi-complete digraph on $n$ vertices and let $w \geq 0$ be an integer.
By~\cref{thm:dpwsemi}, one can compute a directed path-decomposition of width at most $w$, or correctly report that $\dpw(G)>w$ in time $2^{\mathcal{O}(w\log w)}\cdot n^2$.
In the latter case, we directly report that $\crank(G)>w$ by~\cref{prop:dpwcr}.

We may therefore assume that we have obtained a directed path-decomposition of $G$.
Let $k$ denote its width.
Then $k \leq w$.
By~\cref{lem:dpwdcwsemi}, this decomposition can be converted into a  directed $(k+2)$-expression of $G$ in time $\mathcal{O}(n^2)$.
% Let $\mathbf{T}$ be the corresponding $(k+2)$-expression tree of $G$.
% Then $\mathbf{T}$ has at most $\mathcal{O}(k) \cdot n$ nodes by~\cref{lem:dpwdcwsemi}.

We now apply~\cref{thm:main} with $k+2$ labels.
Since $k \leq w$, this decides whether $\crank(G)\le w$ in time $2^{\mathcal{O}((w+1)4^{k+2})}\cdot n=2^{\mathcal{O}((w+1)4^{w})}\cdot n$.
Including the time required to compute the directed path-decomposition and convert it into a directed $(k+2)$-expression, the total running time for deciding whether $\crank(G)\le w$ is
\[2^{\mathcal{O}((w+1)4^{w})}\cdot n+2^{\mathcal{O}(w\log w)}n^2+\mathcal{O}(n^2)\le 2^{\mathcal{O}((w+1)4^{w})}\cdot n^2.\]

For construction, by~\cref{thm:main}, a cycle rank decomposition of $G$ of depth at most $w$ can be constructed from the directed $(k+2)$-expression in time $2^{\mathcal{O}((w+1)4^{w})}\cdot (n+m)=2^{\mathcal{O}((w+1)4^{w})}\cdot n^2$, where $m=\lvert E(G)\rvert =\mathcal{O}(n^2)$.
Including the same preprocessing, the total construction time is also $2^{\mathcal{O}((w+1)4^{w})}\cdot n^2$.
\end{proof}

\section{Minimum Feedback Arc Set parameterized by cycle rank}\label{sec:minfas}

In this section, we consider the problem of finding a minimum feedback arc set in a semi-complete digraph. 
A set $F$ of edges of a digraph $G$ is said to be a \emph{feedback arc set} of $G$ if $G-F$ has no directed cycles.

\vskip 0.3cm
\noindent
\fbox{\parbox{0.97\textwidth}{
	\textsc{Minimum Feedback Arc Set}\\
	\textbf{Input:} A digraph $G$\\
	\textbf{Task:} Compute the minimum size of a feedback arc set of $G$.
    }}
\vskip 0.3cm

The following observation shows that \textsc{Minimum Feedback Arc Set} can be formulated as a linear layout problem on the vertices.
Given a digraph $G$ and a linear order $\sigma = (v_1, \ldots, v_n)$ on $V(G)$, we say an edge $(v_i, v_j) \in E(G)$ is a \emph{$\sigma$-backedge} if $i>j$.
Then for every linear order $\sigma$ on $V(G)$, the set of $\sigma$-backedges is always a feedback arc set of $G$.
Moreover, if $F$ is a minimum feedback arc set of $G$, then a topological ordering of $G-F$ witnesses exactly the removed backedges.
Thus, the objective of \textsc{Minimum Feedback Arc Set} is to find a linear order $\sigma$ with the minimum number of $\sigma$-backedges.
For $S\subseteq V(G)$, a $\sigma$-backedge $(v_i, v_j)$ is called a \emph{$(\sigma,S)$-backedge} if $\{v_i,v_j\}\cap S\neq \varnothing$.
For a set $S$, let $\Sigma(S)$ denote the set of all linear orders on $S$.

\begin{figure}[t]
    \centering
    \begin{tikzpicture}[
        scale=1.3, 
        font=\small,
        dot/.style={circle, fill=black, inner sep=1.5pt, minimum size=3pt}
    ]
        \node[dot, label=left:{$t$}] (t) at (0,0) {};

        \draw[thick] (t.center) -- (-1.2, -2.0) -- (1.2, -2.0) -- cycle;
        \node at (0, -1.2) {$S_t$};

        \draw[dashed, red] (t.center) -- (2.4, 2.4);

        \node[dot] (a1) at (0.6, 0.6) {};
        \node[dot] (a2) at (1.2, 1.2) {};
        \node[dot] (a3) at (1.8, 1.8) {};
        \node[dot] (a4) at (2.4, 2.4) {};

        \draw[thick] (0.5, 0.7) -- (0.35, 0.85) -- (2.15, 2.65) -- (2.3, 2.5);
        \node[sloped, above=2pt] at (1.25, 1.75) {$A_t$};

    \end{tikzpicture}
    \caption{The set $A_t$ of proper ancestors of a node $t$ and the vertex set $S_t$ of the subtree $T_t$ rooted at $t$.}
    \label{fig:structural_decomposition}
\end{figure}
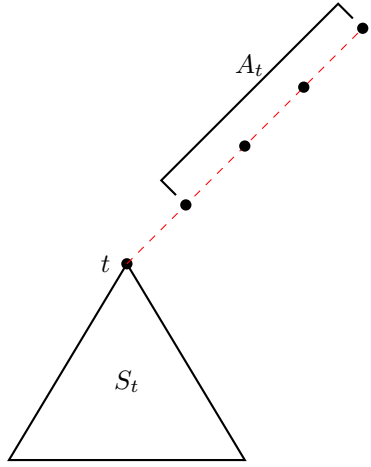

\thmMFAS*
\begin{proof}
Let $G$ be a semi-complete digraph on $n$ vertices and let $w\ge 1$ be the depth of the given cycle rank decomposition of $G$.
We may assume that $G$ is strongly connected, since otherwise, the optimum value of $G$ is equal to the sum of the optimum values of its strongly connected components.
Let $(T,r)$ be the given cycle rank decomposition of $G$ of depth $w$.
% Furthermore, by applying~\Cref{thm:main1}, we can compute a cycle rank decomposition $(T, r)$ of $G$ of depth $w$ in time $\mathcal{O}(9^{(w+1)4^{w+2}} \cdot n^2)$.
% Thus, it remains to show that \textsc{Minimum Feedback Arc Set} can be solved in time $n^{\mathcal{O}(w)}$ when such a decomposition is also given.

For each node $t$ of $T$, let $A_t$ be the set of all proper ancestors of $t$, let $d_t = \abs{A_t}$, and let $S_t=V(T_t)$. 
We note that $A_r=\varnothing$ and $\Sigma(A_r)=\{\varnothing\}$ for the root $r$.
Since $T$ has depth $w$, we have $d_t \le w$ for every node $t$.
Additionally, for a node $t$ in $T$, a linear order $\sigma=(a_1,\ldots, a_{d_t})\in \Sigma(A_t)$, and a tuple $\mathbf{b} = (b_0,\ldots, b_{d_t})$ of non-negative integers with $\sum_{i=0}^{d_t} b_i=\abs{S_t}$, we define
    $\mathsf{DP}[t, \sigma, \mathbf{b}]$
as the minimum number of $(\delta,S_t)$-backedges over all possible linear orders $\delta$ of $S_t\cup A_t$ of the form
\[
    \delta=B_0\oplus \{a_1\}\oplus B_1 \oplus \{a_2\} \oplus B_2\oplus \cdots \oplus \{a_{d_t}\}\oplus B_{d_t},
\]
where $B_i$ is a linear order of a subset of $S_t$ with $\abs{B_i} = b_i$ for each $i \in [0, d_t]$.
As $S_r=V(G)$, the minimum size of a feedback arc set in $G$ is equal to $\DP[r, \varnothing, ( \lvert V(G) \rvert)]$.

For intermediate steps, we also compute the following. Let $t$ be an internal node of $T$ with children $t_1, \ldots, t_s$ such that for every $1 \leq p < q \leq s$, the edges between $S_{t_p}$ and $S_{t_q}$ in $G$ are oriented from $S_{t_p}$ to $S_{t_q}$.
Let $U_j = \bigcup_{i \in [j]} S_{t_i}$ for each $j \in [s]$.
For a linear order $\tau=(a_1,\ldots, a_{d_t+1})\in \Sigma(A_t\cup \{t\})$ and a tuple $\mathbf{x} = (x_0,\ldots, x_{d_t+1})$ of non-negative integers with $\sum_{i \in [0,d_t+1]} x_i=\abs{U_j}$, we define
$\mathsf{DP}'[t,j,\tau, \mathbf{x}]$
as the minimum number of $(\delta, U_j)$-backedges over all linear orders $\delta$ of $U_j \cup A_t\cup \{t\}$ of the form
\[
    \delta=B_0\oplus \{a_1\}\oplus B_1 \oplus \{a_2\} \oplus B_2\oplus \cdots \oplus \{a_{d_t+1}\}\oplus B_{d_t+1}
\]
where $B_j$ is a linear order of a subset of $U_j$ such that $\abs{B_j}=x_j$ for each $j\in [0,d_t+1]$. 
Observe that 
\[  
    \mathsf{DP}'[t,1,\tau, \mathbf{x}]= \mathsf{DP}[t_1,\tau,\mathbf{x}].
\]
For $i\ge 2$, we will recursively compute $\mathsf{DP}'[t,i,\cdot,\cdot]$ by combining the information from $\mathsf{DP}'[t,i-1,\cdot,\cdot]$ and $\mathsf{DP}[t_i,\cdot,\cdot]$.

We now explain how to update the DP tables bottom-up.

\medskip
\noindent
\textbf{Initialization.}
Suppose that $t \in V(T)$ is a leaf node in $T$.
Let $\mathbf{b} = (b_0, \ldots, b_{d_t})$ be a tuple of non-negative integers such that $\sum_{i\in [0,d_t]} b_i = \lvert S_t \rvert = 1$.
Then there is a unique index $h \in [0, d_t]$ such that $b_h=1$.
Thus, for each $\sigma = (a_1, \ldots, a_{d_t}) \in \Sigma(A_t)$, we set $\mathsf{DP}[t, \sigma, \mathbf{b}]$ to the number of $(\sigma', \{t\})$-backedges in $G[\{t\} \cup A_t]$, where 
$\sigma' = (a_1, \ldots, a_h, t, a_{h+1}, \ldots, a_{d_t})$.

\medskip
\noindent
\textbf{Inductive step ($\mathsf{DP}'[t,j,\cdot,\cdot]\to \mathsf{DP}'[t,j+1,\cdot,\cdot]$).}
Let $t$ be an internal node of $T$ with children $t_1, \ldots, t_s$ such that 
for every $1 \leq p < q \leq s$ the edges between $V(T_{t_p})$ and $V(T_{t_q})$ are oriented from $V(T_{t_p})$ to $V(T_{t_q})$. 
Let $j\in [s-1]$. 

Let $\tau=(a_1,\ldots, a_{d_t+1})\in \Sigma(A_t\cup \{t\})$ be a linear order of $A_t\cup \{t\}$.
Suppose that the table for $j$ is computed.
Let $\mathbf{x}=(x_0,\ldots, x_{d_t+1})$ be a tuple of non-negative integers such that $\sum_{i\in [0,d_t+1]}x_i=\abs{U_{j+1}}$.
To compute the transition for $\mathbf{x}$ at step $j+1$, we consider all valid ways to split $\mathbf{x}$ into a tuple $\mathbf{y}$ representing the distribution of $S_{t_{j+1}}$ and a tuple $\mathbf{x}-\mathbf{y}$ representing the distribution of $U_j$.

Specifically, let $\mathbf{y}=(y_0,\ldots, y_{d_t+1})$ be a tuple of non-negative integers such that $0 \leq y_i\le x_i$ for all $i \in [0,d_t+1]$ and $\sum_{i\in [0,d_t+1]}y_i=\abs{S_{t_{j+1}}}$.
Since $G$ is semi-complete and the children are topologically ordered, every edge between $U_j$ and $S_{t_{j+1}}$ is oriented from $U_j$ to $S_{t_{j+1}}$.
Moreover, if some vertices in $U_j$ and some vertices in $S_{t_{j+1}}$ are contained in the same interval divided by $a_1, \ldots, a_{d_t+1}$, then the vertices in $U_j$ can be ordered before the vertices in $S_{t_{j+1}}$ so that no additional backedges between them are created.
Thus, defining
\[\mathsf{penalty}(\mathbf{x},\mathbf{y})\coloneqq \sum_{0\le k<m\le d_t+1}y_k(x_m-y_m),\]
we update the DP table by
\[
    \DP'[t,j+1, \tau, \mathbf{x}]=\min_{\mathbf{y}}\left(\DP'[t,j,\tau,\mathbf{x}-\mathbf{y}]+\DP[t_{j+1},\tau,\mathbf{y}]+\mathsf{penalty}(\mathbf{x},\mathbf{y})\right),
\]
where the minimum is taken over all tuples $\mathbf{y}$ satisfying the above conditions.

\medskip
\noindent
\textbf{Finalizing the table for $t$ ($\mathsf{DP}'[t,s,\cdot,\cdot]\to \mathsf{DP}[t,\cdot,\cdot]$).}
After computing the tables for $j=s$, we finalize the table for the node $t$.
For each linear order $\sigma=(a_1,\ldots, a_{d_t}) \in \Sigma(A_t)$ and $h \in [0,d_t]$, let
\[
    \sigma_h\coloneqq \begin{cases}
        (t,a_1,\ldots,a_{d_t}),&\text{if $h=0$,}\\
        (a_1,\ldots, a_h,t,a_{h+1},\ldots, a_{d_t}),&\text{if $0<h<d_t$, and}\\
        (a_1,\ldots,a_{d_t},t),&\text{if $h=d_t$.}
    \end{cases}
\]
Additionally, let $c(t,\sigma,h)$ be the number of $(\sigma_h,\{t\})$-backedges, that is,
\[c(t,\sigma,h)=\abs{\{i\in [h]:\text{$(t,a_i)\in E(G)$}\}}+\abs{\{i\in [h+1,d_t]: \text{$(a_i,t)\in E(G)$}\}}.\]
Now, for each tuple $\mathbf{b}=(b_0,\ldots, b_{d_t})$ of non-negative integers with $\sum_{i\in [0,d_t]} b_i=\abs{S_t}$, let $\Pi_{h, \mathbf{b}}$ be the set of all tuples $\mathbf{b}' = (b_0', \ldots, b_{d_t+1}')$ such that 
\begin{itemize}
    \item $b_j' = b_j$ for every $j \in [0, h-1]$, 
    \item $b_h' + b_{h+1}' = b_h-1$, and
    \item $b_{j+1}'=b_j$ for every $j \in [h+1, d_t]$.
\end{itemize}
Then we take
\[
    \DP[t,\sigma,\mathbf{b}]=\min_{h\in [0,d_t]}\left(c(t,\sigma,h)+\min_{\mathbf{b}'\in\Pi_{h,\mathbf{b}}}\DP'[t,s,\sigma_h,\mathbf{b'}]\right).
\]

\medskip
\noindent
\textbf{Correctness.}
To show the correctness of the algorithm, we need to show that the recurrence precisely computes the minimum number of backedges.
We use bottom-up induction on the cycle rank decomposition tree $T$, where for every internal node $t$ with ordered children $t_1,\ldots, t_s$, the induction step processes $t_1,\ldots, t_s$ in this order before processing $t$.

For the base case, we assume that $t$ is a leaf node of $T$.
We note that $S_t=\{t\}$.
By iterating over all possible positions $j\in [0,d_t]$ for $t$ relative to each linear order $\sigma\in \Sigma(A_t)$, the algorithm considers all linear orders of $\{t\}\cup A_t$.
As we assigned, the value $\DP[t,\sigma,\mathbf{b}]$ is the exact number of $(\sigma',\{t\})$-backedges.

For the inductive step, we may assume that $t$ is an internal node of $T$ with children $t_1,\ldots, t_s$, and assume that each $\DP[t_i,\cdot,\cdot]$ is computed by induction.
To prove that the DP table $\DP[t,\sigma,\mathbf{b}]$ correctly computes the minimum number of backedges for the subtree $T_t$, it suffices to show that the sequential merging processes of its children are computed correctly.
This follows from the following claim.

\begin{claim}\label{claim:FAScorrect}
For each internal node $t\in V(T)$ with children $t_1,\ldots, t_s$, each index $j\in [s]$, each linear order $\sigma=(a_1,\ldots, a_{d_t+1})\in \Sigma(A_t\cup\{t\})$, each tuple $\mathbf{x}=(x_0,\ldots, x_{d_t+1})$, the value $\DP'[t,j,\sigma,\mathbf{x}]$ is equal to the minimum number of $(\delta,\bigcup_{i\in [j]}S_{t_i})$-backedges over all linear orders $\delta$ of $A_t\cup \{t\}\cup \left(\bigcup_{i\in [j]}S_{t_i}\right)$ of the form
\[\delta=X_0\oplus \{a_1\}\oplus X_1\oplus \cdots \oplus X_{d_t}\oplus \{a_{d_t+1}\}\oplus X_{d_t+1}\]
such that $\abs{X_k}=x_k$ for each $k\in [0,d_t+1]$.
\end{claim}
\begin{clproof}
For each $j\in [s]$, let $\mathcal{D}_j(\sigma,\mathbf{x})$ be the set of all linear orders $\delta$ of $A_t\cup \{t\}\cup \left(\bigcup_{i\in [j]}S_{t_i}\right)$ which can be decomposed as $\delta=X_0\oplus \{a_1\}\oplus X_1\oplus \cdots \oplus X_{d_t}\oplus \{a_{d_t+1}\}\oplus X_{d_t+1}$ such that $\abs{X_k}=x_k$ for each $k\in [0,d_t+1]$, and let $\mathcal{F}_j(\sigma,\mathbf{x})$ be the set of all linear orders $\delta$ of $A_t\cup \{t\}\cup S_{t_j}$ which can be decomposed as $\delta=X_0\oplus \{a_1\}\oplus X_1\oplus \cdots \oplus X_{d_t}\oplus \{a_{d_t+1}\}\oplus X_{d_t+1}$ such that $\abs{X_k}=x_k$ for each $k\in [0,d_t+1]$.

We use induction on $j$.
The base case $j=1$ holds, since $\DP'[t,1,\sigma,\mathbf{x}]=\DP[t_1,\sigma,\mathbf{x}]$ by definition.

We may assume that $j\ge 2$ and the claim holds for $j-1$.
Let $\mathsf{cost}(\delta)$ be the number of $(\delta,\bigcup_{i\in [j]}S_{t_i})$-backedges.
We claim that $\DP'[t,j,\sigma,\mathbf{x}]=\min_{\delta\in \mathcal{D}_j(\sigma,\mathbf{x})}\mathsf{cost}(\delta)$.

First, we show that $\DP'[t,j,\sigma,\mathbf{x}]\ge \min_{\delta\in \mathcal{D}_j(\sigma,\mathbf{x})}\mathsf{cost}(\delta)$.
Let $\mathbf{y}$ be the tuple achieving the minimum in the recurrence for $\DP'[t,j,\sigma,\mathbf{x}]$.
By induction, there is a linear order $\delta_1\in \mathcal{D}_{j-1}(\sigma,\mathbf{x}-\mathbf{y})$ with blocks $Z_k$ where $\abs{Z_k}=x_k-y_k$ for each $k$ and a linear order $\delta_2\in \mathcal{F}_j(\sigma,\mathbf{y})$ with blocks $Y_k$ where $\abs{Y_k}=y_k$ for each $k$, such that
\[\mathsf{cost}(\delta_1)=\DP'[t,j-1,\sigma,\mathbf{x}-\mathbf{y}]\quad\text{ and }\quad\mathsf{cost}(\delta_2)=\DP[t_j,\sigma,\mathbf{y}].\]
We construct a linear order $\delta^\ast$ by defining its $k$-th block as $X_k=Z_k\oplus Y_k$ for each $k\in [0,d_t+1]$.
This places all vertices of $\bigcup_{i\in [j]}S_{t_i}$.
Since $G$ is semi-complete and the children are topologically ordered, all edges between $Z_k$ and $Y_k$ are directed from $Z_k$ to $Y_k$.
Thus, each block of $\delta^\ast$ has no backedges between $Y_k$ and $Z_k$.
The only new backedges are edges where a vertex of $Y_k$ appears before a vertex of $Z_m$ with $k<m$.
The number of all such edges is $\sum_{0\le k<m\le d_t+1}\abs{Y_k}\abs{Z_m}=\mathsf{penalty}(\mathbf{x},\mathbf{y})$.
Since $\delta^\ast\in\mathcal{D}_j(\sigma,\mathbf{x})$, we have $\min_{\delta\in\mathcal{D}_j(\sigma,\mathbf{x})}\mathsf{cost}(\delta)\le \mathsf{cost}(\delta^\ast)=\DP'[t,j,\sigma,\mathbf{x}]$.

Now, we show that $\DP'[t,j,\sigma,\mathbf{x}]\le \min_{\delta\in \mathcal{D}_j(\sigma,\mathbf{x})}\mathsf{cost}(\delta)$.
Let $\delta'\in\mathcal{D}_j(\sigma,\mathbf{x})$ be an optimal linear order with blocks $X_k$, which minimizes the number of backedges.
For a sequence $S$ and a subset of vertices $W$, let $S\vert_W$ be the restriction of $S$ to $W$.
For each $k\in [0,d_t+1]$, let $Y'_k \coloneqq X_k\vert_{S_{t_j}}$ and $Z'_k \coloneqq X_k\vert_{\bigcup_{i\in [j-1]}S_{t_i}}$.
Let $y'_k\coloneqq \abs{Y'_k}$ and define the tuple $\mathbf{y}'=(y'_0,\ldots, y'_{d_t+1})$.
Then we have $\abs{Z'_k}=x_k-y'_k$ for each $k$.
Restricting $\delta'$ to $A_t\cup \{t\}\cup \left(\bigcup_{i\in [j-1]}S_{t_i}\right)$ constructs a linear order $\delta'_1\in \mathcal{D}_{j-1}(\sigma, \mathbf{x}-\mathbf{y}')$.
Similarly, restricting $\delta'$ to $A_t\cup \{t\}\cup S_{t_j}$ constructs a linear order $\delta'_2\in \mathcal{F}_{j}(\sigma, \mathbf{y}')$.
By induction, we have $\mathsf{cost}(\delta'_1)\ge \DP'[t,j-1,\sigma,\mathbf{x}-\mathbf{y}']$ and $\mathsf{cost}(\delta'_2)\ge \DP[t_j,\sigma,\mathbf{y}']$.
Moreover, the number of $\delta'$-backedges between $S_{t_j}$ and $\bigcup_{i \in [j-1]} S_{t_i}$ is exactly $\sum_{0\le k<m\le d_t+1}y'_k(x_m-y'_m)=\mathsf{penalty}(\mathbf{x},\mathbf{y}')$.
Hence, we have
\begin{align*}
    \mathsf{cost}(\delta')&\ge \mathsf{cost}(\delta'_1)+\mathsf{cost}(\delta'_2)+\mathsf{penalty}(\mathbf{x},\mathbf{y}')\\
    &\ge \DP'[t,j-1,\sigma,\mathbf{x}-\mathbf{y}']+\DP[t_j,\sigma,\mathbf{y}']+\mathsf{penalty}(\mathbf{x},\mathbf{y}')\\
    &\ge \min_{\mathbf{y''}}\left(\DP'[t,j-1,\sigma,\mathbf{x}-\mathbf{y}'']+\DP[t_j,\sigma,\mathbf{y}'']+\mathsf{penalty}(\mathbf{x},\mathbf{y}'')\right)\\
    &=\DP'[t,j,\sigma,\mathbf{x}].
\end{align*}
This proves the claim.
\end{clproof}

By \cref{claim:FAScorrect}, the intermediate DP table $\DP'[t,s,\sigma_h,\mathbf{b}']$ correctly computes the minimum number of backedges induced by the union of $S_t\setminus \{t\}$ and $A_t\cup\{t\}$.
In the finalization step for the internal node $t$, the algorithm iterates over all possible positions $h\in [0,d_t]$ for $t$ relative to a linear order $\sigma\in\Sigma(A_t)$.
The term $c(t,\sigma,h)$ correctly counts all the backedges between $t$ and its proper ancestors.
The set $\Pi_{h,\mathbf{b}}$ correctly collects every tuple $\mathbf{b}'$ of the children which account for the specific placement of $t$ at position $h$ such that they project down to the tuple $\mathbf{b}$.
Taking the minimum over all choices of $h$ guarantees that $\DP[t,\sigma,\mathbf{b}]$ records the minimum number of backedges for $S_t$.

By induction on $T$, the DP tables for all nodes store the minimum number of backedges correctly.
At the root node $r$, the ancestor set is empty.
This implies $d_r=0$ and $S_r=V(G)$.
Hence, the single entry $\DP[r,\varnothing, (\abs{V(G)})]$ computes the minimum size of a feedback arc set of $G$.

\medskip
\noindent
\textbf{Time complexity.}
Since the depth of $(T, r)$ is $w$, we have $d_t \leq w$, $\abs{\Sigma(A_t)} \leq w!$, and $\abs{\Sigma(A_t \cup \{t\})} \leq (w+1)!$ for every node $t$ in $T$.
Furthermore, the number of tuples consisting of at most $w+2$ non-negative integers that sum to at most $n$ is $n^{\mathcal{O}(w)}$.
Thus, the number of states over the whole algorithm is~$n^{\mathcal{O}(w)}$.

For a leaf node, the table $\DP$ can be filled directly in
time $n^{\mathcal{O}(w)}$.
For an internal node $t$, it has at most $n$ children, and hence there are at most $n$ transitions for $t$ to compute the table $\DP'$.
In each state of such a transition, we enumerate at most $n^{\mathcal{O}(w)}$ choices of $\mathbf{y}$, and for each such $\mathbf{y}$, the value of $\mathsf{penalty}(\mathbf{x}, \mathbf{y})$ can be computed in time $\mathcal{O}(w^2)$.
Thus, it takes time $n^{\mathcal{O}(w)}$ to compute all tables $\DP'[t, \cdot, \cdot, \cdot]$.
During the finalization step, each state enumerates at most $w+1$ choices of $h$, and for each $h$, there are at most $n+1$ choices for the choice of the tuple $\mathbf{b}'$.
This shows that the finalization step also runs in time $n^{\mathcal{O}(w)}$.
Consequently, the dynamic programming runs in time~$n^{\mathcal{O}(w)}$.
\end{proof}

\section{Conclusion}\label{sec:conclusion}

In this paper, we proved that \textsc{Cycle Rank} is fixed-parameter tractable parameterized by $w$ and directed clique-width.
To achieve this, we introduced the notion of CW-rankings, which is equivalent to the cycle rank and allows us to track the reachability conditions of directed closed walks.
By applying this result, we further provided a fixed-parameter tractable algorithm parameterized solely by $w$ for the class of semi-complete digraphs.

While \cref{thm:main} establishes fixed-parameter tractability, its running time relies on a double-exponential dependence on the parameters $w$ and $k$.
This naturally raises the question of whether the parameter dependence of our algorithm in \cref{thm:main} can be improved.
\begin{question}
Does there exist an algorithm for \textsc{Cycle Rank} on digraphs of directed clique-width $k$ that runs in time $2^{f(k,w)}\cdot n^{\mathcal{O}(1)}$ for some polynomial function $f$?
\end{question}

Since every decidable fixed-parameter tractable problem admits a kernel of some computable size, \cref{thm:main1} gives a computable kernel.
It remains natural to ask whether a polynomial kernel exists.

\begin{question}
Does \textsc{Cycle Rank} parameterized by $w$ admit a polynomial kernel when the input digraph is semi-complete?
\end{question}

We show that \textsc{Minimum Feedback Arc Set} on semi-complete digraphs is slicewise polynomial parameterized by the cycle rank.
A natural next step is to ask whether the parameter dependence can be improved to fixed-parameter tractability. By \Cref{thm:main1}, we may assume that the cycle rank decomposition is given as input.
\begin{question}
Is \textsc{Minimum Feedback Arc Set} on semi-complete digraphs fixed-parameter tractable parameterized by the cycle rank of the input digraph?
\end{question}

Finally, a major open question, originally posed by Gruber~\cite{Gruber2012} and by Giannopoulou, Hunter, and Thilikos~\cite{GIANNOPOULOU2012searchinggame}, remains unsolved.
\begin{question}
Is \textsc{Cycle Rank} fixed-parameter tractable parameterized solely by $w$ on general digraphs?
\end{question}

\bibliographystyle{plain}
\bibliography{tournament_cyclerank}

\end{document}